\documentclass[prx,twocolumn,longbibliography]{revtex4-2}

\usepackage[utf8]{inputenc}
\usepackage{amsmath,amssymb,amsthm,mathrsfs,bbm,bm,mathtools}
\usepackage{graphicx,color, tikz, hyperref}
\usepackage{braket}
\usepackage{dsfont}
\usepackage[titletoc,title]{appendix}

  \theoremstyle{plain}
  \newtheorem{theorem}{Theorem}
  \newtheorem{lemma}[theorem]{Lemma}

  \theoremstyle{definition}
  \newtheorem{definition}[theorem]{Definition}

\newcommand{\abs}[1]{\left\lvert#1\right\rvert}
\newcommand{\norm}[1]{\left\lVert#1\right\rVert}

\begin{document}

\title[Quantum signal processing without angle finding]
{Quantum signal processing without angle finding}
\author{Abhijeet Alase}
\affiliation{Centre for Engineered Quantum Systems, School of Physics,
The University of Sydney, Sydney, New South Wales 2006, Australia}
\email{alase.abhijeet@gmail.com}

\date{\today}

\begin{abstract}
Quantum signal processing (QSP) has emerged as a unifying subroutine in quantum algorithms. 
In QSP, we are given a function $f$ and a unitary black-box $U$,
and the goal is to construct a quantum circuit for 
implementing $f(U)$ to a given precision. 
The existing approaches to performing QSP require a classical 
preprocessing step to compute rotation angle parameters for quantum circuits 
that implement $f$ approximately. However, this classical computation 
often becomes a bottleneck, limiting the scalability and practicality of QSP. 
In this work, we propose a novel approach to QSP that bypasses the 
computationally intensive angle-finding step. Our method leverages 
a quantum circuit for implementing a diagonal operator that encodes $f$, 
which can be constructed from a classical circuit 
for evaluating $f$. This approach to QSP simplifies the circuit 
design significantly while enabling nearly optimal implementation of 
functions of block-encoded Hermitian matrices for black-box functions.
Our circuit closely resembles the phase estimation-based circuit for 
function implementation, challenging conventional skepticism about 
its efficiency. By reducing classical overhead, our work significantly 
broadens the applicability of QSP in quantum computing.


\end{abstract}

\maketitle

\section{Introduction}
Quantum signal processing (QSP) subroutine~\cite{LYC16,LC17,LC19} underpins 
the state-of-the-art algorithms for performing many quantum computational tasks,
including Hamiltonian simulation, solving a system of linear equations, 
solving differential equations and optimization~\cite{MRTC21}. 
In the block encoding framework~\cite{LC19,GSLW19}, QSP leads to an efficient
strategy for implementing functions of Hermitian matrices and for
performing quantum singular value transformation (QSVT). 
QSVT powered by QSP has 
proved crucial in developing quantum algorithms with optimal or 
exponentially improved dependence on precision for the aforementioned tasks. 
In QSP, we are given a function $f$ and
a black-box unitary $U$, and the goal is to construct a quantum circuit 
to implement $f(U)$ up to
a given precision $\epsilon$. The quantum circuits proposed in 
the literature~\cite{LC17,GSLW19,MW24}
for QSP are extremely efficient in their use of the quantum resources, namely
queries to $U$, $U^\dagger$ and additional one- and two-qubit gates. 
However, in these approaches, one first needs
to classically calculate certain rotation angles that parametrize 
the quantum circuit for implementing a Laurent polynomial 
approximation to $f$~\cite{Haa19,CDGH+20}. 
This classical step has been identified as a bottleneck in the existing approaches to QSP,
and may limit the utility of QSP in practice~\cite{Haa19,CDGH+20,DMWL21,YY24,DLNW24}.
In this paper, we construct a circuit for QSP that 
bypasses the polynomial approximation and angle-finding steps,
thereby significantly simplifying the design of the circuit
and reducing the classical cost of QSP.

Let us begin by reviewing one of the existing approaches for QSP~\cite{MW24,BMPW24}. The function
$f: S^1 \to \mathbb{C}$ on the unit circle is either fixed, or it is assumed to be
provided via coefficients of some series expansion, such as Taylor series
or Chebyshev series. Using these coefficients, one classically computes
a Laurent polynomial $f_d \in \mathbb{C}_d[z,z^{-1}]$ of 
degree $d$ approximating $f$ on the unit circle. 
The degree $d$ is chosen such
that the approximation error 
\begin{equation}
\norm{f(z)-f_d(z)}_\infty := \max_{z \in S^1}\abs{f(z)-f_d(z)}
\end{equation}
is within the desired precision $\epsilon$. In the
next step, an alternating sequence of the directionally controlled operation 
$W = \ket{0}\bra{0}\otimes U+ \ket{1}\bra{1}\otimes U^\dagger$ 
and single-qubit rotations on the first (control) qubit is employed to implement $f_d(U)$. 
The angles specifying single-qubit rotations around the $x$- and the $z$-axes 
are the parameters of this quantum circuit. Assuming $\norm{f_d}_\infty \le 1$,
this can be achieved with $d$ uses of $W$ and $d$ composite 
single-qubit rotations.

Calculation of the single-qubit rotation angles parametrizing the QSP circuit
is known to be challenging and has attracted 
significant attention over the last few years~\cite{GSLW19,Haa19,CDGH+20,DMWL21,YY24,DLNW24}.
While these angles can be computed classically by some deterministic 
algorithms, their complexity scales unfavourably as ${\bf O}(d^2)$ or 
worse~\cite{GSLW19,DMWL21}. Apart from
their high cost, these deterministic algorithms are 
known to be numerically unstable~\cite{DMWL21}.
Alternatively, optimization algorithms can be used to compute these angles
approximately~\cite{CDGH+20,DMWL21,YY24,DLNW24}, 
but a provable bound on the asymptotic complexity of such 
algorithms is not known to our knowledge. 

Apart from angle-finding, the classical cost of computing a Laurent polynomial
approximation $f_d$ itself can be comparable or higher than the cost of computing 
the rotation angles~\cite{GL17,GSLW19,AGGW20}. For achieving the best query complexity, 
one aims to obtain the best Laurent polynomial approximation $f_{d,*}(x)$ to $f(x)$. While
some iterative algorithms are known for computing the
best Laurent polynomial approximation~\cite{Fra65}, the convergence rate of such algorithms for
general $f$ is not known, thereby obstructing a rigorous upper bound
on the classical cost of computing such an approximation.
Another approach is to use truncated or interpolated Fourier series to obtain
a Laurent polynomial approximation, which could result in a suboptimal dependence 
of the query complexity on precision in general~\cite{Wan21}.

In this paper, we use ideas from the theory of interpolation~\cite{Ste05} to construct
a new circuit for QSP. The design of this circuit is extremely simple
and bypasses both the classical steps required 
by the existing approaches to QSP,
namely finding a Laurent polynomial approximation as well as
finding rotation angles to implement that polynomial. For this purpose, 
we assume access to a diagonal unitary encoding of $f$~\cite{ZNSD24}. 
A similar encoding was used for implementing matrix inverse
in the seminal HHL algorithm~\cite{HHL09}.
Such an encoding can always be constructed from a classical
circuit for computing $f$ (see Sec.~\ref{sec:qsp}
for further discussion).
Several methods for constructing such an encoding are known~\cite{ZNSD24}.
Our circuit for QSP, which uses this encoding of $f$, 
can be constructed exponentially faster on a classical computer,
and has the same query complexity as the existing approaches up to
a constant prefactor. Interestingly, our circuit makes only one use of
the operator encoding $f$. 
We also show the application of our circuit for QSP to 
implement functions of Hermitian matrices and implementing 
quantum singular value transformation (QSVT)~\cite{GSLW19}.

There are two main implications of our work. First, our circuit for
QSP, and therefore for implementing functions of Hermitian matrices,
can be constructed efficiently and stably on a classical computer
for any efficiently computable $f$. This overcomes the key limitation of
the existing approaches to QSP, without compromising the scaling
of the query complexity. Second, our circuit for QSP is remarkably
similar to the phase estimation-based circuit, for instance the one used
in HHL algorithm for matrix inversion. Prior to our work, this approach was considered expensive
 due to unfavourable complexity of the phase estimation subroutine.
Our work, in effect, provides a different way of analyzing errors in the phase estimation-based
approach for implementing functions of matrices. Thus, it is seen that the accuracy of the implemented
function is often, to our surprise, significantly better than the accuracy to which the
eigenphases are estimated in this approach.

The organization of the rest of the paper is as follows. In Sec.~\ref{sec:results},
we provide an overview of the main results. In Sec.~\ref{sec:qsp}, we prove the correctness
and complexity of our circuit for QSP. In Sec.~\ref{sec:qevt}, we discuss the application
of this circuit for performing functions of block-encoded Hermitian matrices and
singular value transformation of arbitrary block-encoded matrices. 
Finally, in Sec.~\ref{sec:conclusion}, we summarize the results and discuss open problems.

\section{Results and discussion}
\label{sec:results}
To state the main results of the paper, we need to introduce three 
definitions. First, let us define the diagonal encoding of 
a given function $f$.
\begin{definition}[Diagonal encoding of a function]
Let $f:S^1 \to \mathbb{C}$ be such that $\norm{f}_{\infty} \le 1$. 
For $d=2^m$ with $m \in \mathbb{Z}^+$, we say that
an $(m+3)$-qubit unitary operator $U_{f,4d}$ encodes $f$ if
\begin{equation}
\label{eq:funitaryencoding}
    \bra{j'}\bra{0}U_{f,4d}\ket{j}\ket{0} = \delta_{jj'}f\left(e^{i 2\pi j/4d}\right),\quad 
    j,j' \in [4d],
\end{equation}
where $[4d] := \{0,1,\dots,4d-1\}$. 
\end{definition}

Next, we define block encoding of a square matrix.
\begin{definition}[Block encoding~\cite{GSLW19}]
For a $2^n \times 2^n$ matrix $A$ and $\alpha_A,\epsilon_A \in \mathbb{R}^+$,
we say that 
an $(n+a_A)$-qubit unitary operator $U_A$ is a 
$(\alpha_A,a_A,\epsilon_A)$-block encoding of $A$ if 
\begin{equation}
\abs{A - \alpha_A\bra{0^a}U_A\ket{0^a}} \le \epsilon_A.
\end{equation}
\end{definition}

Finally, we define the best Laurent polynomial approximation to a function.
\begin{definition}[Best Laurent polynomial approximation~\cite{Ste05}]
For a function $f:S^1 \to \mathbb{C}$ on the unit circle, 
we say that $f_{d,*}\in\mathbb{C}_d[z,z^{-1}]$ 
is the best degree-$d$ Laurent polynomial approximation 
to $f$ if 
\begin{equation}
	\norm{f-f_{d,*}}_{\infty} = \min_{g \in \mathbb{C}_d[z,z^{-1}]} \norm{f-g}_{\infty}.
\end{equation}
The best approximation error is defined to be
\begin{equation}
    E_d(f):= \norm{f-f_{d,*}}_{\infty}.
\end{equation}
\end{definition}
Note that $E_d(f)$ can also be interpreted as the error
in the best trigonometric polynomial approximation to 
the function $\tilde{f}:\mathbb{R} \to \mathbb{C}$ 
defined by $\tilde{f}(\theta) = f(e^{i\theta})$. This follows 
from the fact that $\tilde{f}_{d,*}:\mathbb{R} \to \mathbb{C}$
defined by $\tilde{f}_{d,*}(\theta) = f_{d,*}(e^{i\theta})$ is a degree-$d$ 
trigonometric polynomial in $\theta$, and therefore it is the best 
trigonometric polynomial approximation to $\tilde{f}(\theta)$. 

We are now ready to state the main result of this paper.
\begin{theorem}[QSP]
\label{thm:newQSP}
Given an $n$-qubit black-box unitary $U$ 
and a unitary $U_{f,4d}$ for $d = 2^m$, $m \in \mathbb{Z}^+$ encoding 
a function $f:S^1\to\mathbb{C}$ with $\norm{f}_\infty \le 1$, 
we can construct a $(\sqrt{2},m+3,(1+\sqrt{2})E_d(f))$-block 
encoding of $f(U)$ by making $4d-1$ uses
of c-$U$ and c-$U^\dagger$ each, one use 
of $U_{f,4d}$, and ${\bf O}(m^2)$ additional single- and two-qubit gates. 
Moreover, the circuit description for this block encoding can
be computed in $\text{poly}(m)$ time on a classical computer.
\end{theorem}

The proof of Theorem~\ref{thm:newQSP} is included in Sec.~\ref{sec:qsp}.
Although $d$ is just a parameter in Theroem~\ref{thm:newQSP}, 
it is convenient to think of $d$ as the degree of the approximating polynomial
for the purpose of comparison to the conventional QSP circuit.
Notice that the classical cost in Theorem~\ref{thm:newQSP}
scales as ${\bf O}(m^2) \in {\bf O}(\log^2{d})$. 
If the classical cost of constructing a circuit for $U_{f,4d}$ 
is in $\text{polylog}(d)$, then the overall classical cost of
our algorithm is also $\text{polylog}(d)$. In this case, our circuit
provides an exponential improvement in the classical cost
over the circuits proposed for QSP in the literature, which require
${\bf O}(d^3)$ classical operations for angle-finding alone~\cite{DLNW24}, and 
more for computing a degree-$d$ polynomial approximation for $f$~\cite{AGGW20}. 
More importantly, our circuit for QSP is extremely simple and
is uniquely determined by the parameter $d$,
and therefore overcomes the issue of numerical instabilities encountered by the
conventional approaches to QSP~\cite{DLNW24}. 

If one uses the best Laurent polynomial approximation of degree $d$, 
which is demanding to compute, then the error for conventional QSP scales 
as $E_d(f)$. In comparison, the error for our 
circuit for QSP scales as $(1+\sqrt{2})E_d(f)$, which is worse only by
a constant factor. However, 
the conventional approach to QSP makes many 
queries to the (classical) function oracle for the purpose of computing the best
Laurent polynomial approximation, and the queries grow with the required precision. 
Our circuit makes only one query to $U_{f,4d}$, which in turn can be constructed
using two queries to the binary quantum oracle for $f$ 
as we discuss in Sec.~\ref{sec:qsp}. 
These desirable features of our new circuit for QSP come at the cost of 
only a constant overhead for queries to $U$ compared to the 
conventional QSP circuit (${\bf O}(d)$ instead of $d$). Therefore, we  
achieve a circuit for QSP that makes nearly optimal uses of $U$, $U_{f,4d}$ and
additional two-qubit gates. 

In Sec.~\ref{sec:qevt}, we study the application of our new circuit for QSP
for implementing functions of Hermitian matrices (FHM)~\cite{GSLW19,CGJ19,SSJ19,AGGW20}. 
Our main result is the following theorem.
\begin{theorem}[FHM]
\label{thm:newFHM}
Let $g:[-1,1]\to\mathbb{C}$ with $\abs{g(x)} \le 1$
and let $f:S^1\to\mathbb{C}$ be defined by
$f(e^{i\theta}) = g(\cos{\theta})$ for $\theta \in \mathbb{R}$.
Given an $(n+a)$-qubit block encoding $U_H$ for an $n$-qubit
Hamiltonian $H$ and a unitary $U_{f,4d}$ 
for $d = 2^m$, $m \in \mathbb{Z}^+$, 
we can construct a 
$(\sqrt{2},a+m+3,(1+\sqrt{2})E_d(f))$-block 
encoding $U_{g(H)}$ of $g(H)$ by making $4d-1$ uses
of c-$U_H$ and c-$U_H^\dagger$ each, one use of $U_{f,4d}$, and 
${\bf O}(da)$ additional two-qubit gates. 
Moreover, a circuit description of $U_{g(H)}$ can
be computed using $\text{poly}(m)$ classical operations.
\end{theorem}
Similar to QSP, the query complexity of FHM stated in Theorem~\ref{thm:newFHM} is equivalent to
the best results in the literature up to a constant prefactor, but the classical 
cost is exponentially smaller. 
Moreover, we discuss how Jackson's inequalities
can be used to bound the error $(1+\sqrt{2})E_d(f)$ in the resulting block encoding.
As an application, we show that for a given non-integer positive number $c$, 
a block encoding of $\abs{H}^c$ to precision $\epsilon$ can be constructed
using ${\bf O}(1/\epsilon^{1/c})$ queries to $U_H$ and $U_H^\dagger$. 
We also discuss how analytic and entire functions can be handled
within this approach. Finally, we discuss how Theorem~\ref{thm:newFHM}
can be used to perform quantum singular value transformation (QSVT),
achieving equivalent query complexity as the best results in the literature
but with a circuit that is easy to construct classically.

\section{Quantum signal processing by polynomial interpolation}
\label{sec:qsp}
In this section, we prove Theorem~\ref{thm:newQSP} by constructing a block encoding of
$f(U)$. Let us first address the special case when $f$ is a degree-$d$ Laurent polynomial. 
Our strategy is to make use of interpolation by Laurent polynomial.
We use of the following expression for $f(z)$.
\begin{lemma}
\label{lem:interpolation1}
Let $f \in \mathbb{C}[z,z^{-1}]$ be a degree-$d$ Laurent polynomial
and $d_1>2d$ an integer. Then
for $\{z_k = \exp\left(2\pi ik/d_1\right), \ k=0,\dots,d_1-1\}$, 
\begin{equation}
\label{eq:interpolation1}
    f(z) =z^{-d}\frac{1}{d_1}\sum_{j,k=0}^{d_1-1}f(z_k)z_k^d\left(\frac{z}{z_k}\right)^{j},
\end{equation}
\end{lemma}

\begin{proof}
Being a degree-$d$ Laurent polynomial,
$f(z)$ can be expressed as 
\begin{equation}
\label{eq:laurentpoly}
f(z) = \sum_{j=-d}^{d} \beta_j z^j.
\end{equation} 
By choosing the interpolation points
to be $\{z_k\}$, 
we can express the coefficients of $f$ as
$\beta_j = (\sum_{k=0}^{d_1-1}f(z_k)z_k^{-j})/d_1$. 
Substituting this expression for the coefficients in Eq.~\eqref{eq:laurentpoly}
yields Eq.~\eqref{eq:interpolation1}.
\end{proof}

\begin{figure*}
\begin{tikzpicture}
\node at (-0.5,0){\includegraphics[width = 2\columnwidth]{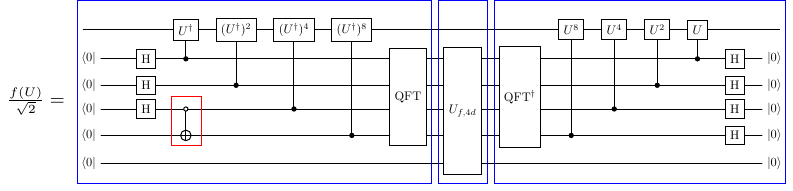}};
\end{tikzpicture}
    \caption{Circuit diagram (in the Kitaev notation) for block encoding 
    of $f(U)$, where $f$ is a degree-$d$ Laurent polynomial (here $d=4$). 
    The circuit acts on the signal register 
    (the top wire) and $3+\lceil \log_2(d)\rceil = 5$ ancilla qubits (all except the top wire). The three parts of the circuit 
    outlined by the three blue boxes can be
    identified as, from right to left, the phase estimation circuit, 
    controlled rotation specified by $f$, and a modified inverse of the phase estimation
    circuit, respectively. The modification refers to the replacement of a Hadamard
    gate by a zero-controlled NOT gate, outlined by the red box. 
    }
    \label{fig:QSP}
\end{figure*}

It is instructive to consider a block encoding 
of $f(U)$ based on this formula.
\begin{lemma}
\label{lem:blockencoding1}
Let $f \in \mathbb{C}[z,z^{-1}]$ be a degree-$d$ Laurent polynomial
and $d_1>2d$ an integer. Let $U$ be an $n$-qubit unitary operator and
\begin{align}
    \mathcal{W}_U &= \sum_{j=0}^{d_1-1}U^j \otimes \ket{j}\bra{j},
    \quad \ket{+_{d_1}} = \frac{1}{\sqrt{d_1}}\sum_{j=0}^{d_1-1}\ket{j}, \nonumber\\
    V_{d_1} &= \sum_{k=0}^{d_1-1}z_k \ket{\varphi_k}\bra{\varphi_k}, \quad    \ket{\varphi_k} = \frac{1}{\sqrt{d_1}}\sum_{j=0}^{d_1-1}z_k^j\ket{j}.
\end{align}
Then
\begin{equation}
\label{eq:circuit1}
        \frac{f(U)}{\sqrt{d_1}}= (U^\dagger)^d\braket{0|[\mathds{1}\otimes 
        f(V_{d_1})V^d_{d_1}] 
        \mathcal{W}_U|+_{d_1}}.
\end{equation}
\end{lemma}

\begin{proof}
It suffices to prove that Eq.~\eqref{eq:circuit1} holds in each eigenspace of $U$, i.e.
\begin{equation}
\label{eq:qsp1}
        \frac{f(z)}{\sqrt{d_1}}= z^{-d}\braket{0|f(V_{d_1})V^d_{d_1} \mathcal{W}_z|+_{d_1}}
\end{equation}
where 
$\mathcal{W}_z = \sum_{j=0}^{d_1-1}z^j\ket{j}\bra{j}$ and $\abs{z} = 1$.
The statement of the lemma then follows by linearity.
Eq.~\eqref{eq:qsp1} follows from
\begin{align}
z^{-d}\bra{0}&f(V_{d_1})V^d_{d_1} \mathcal{W}_z\ket{+_{d_1}}\nonumber \\ 
&= 
\frac{z^{-d}}{\sqrt{d_1}}\sum_{j=0}^{d_1-1}\braket{0|f(V_{d_1})V^d_{d_1} z^j |j} \nonumber \\
&= \frac{z^{-d}}{\sqrt{d_1}}\sum_{j=0}^{d_1-1}\sum_{k=0}^{d_1-1}\bra{0}f(z_k)z_k^d \ket{\varphi_k}\bra{\varphi_k}z^j
\ket{j} \nonumber \\
&= \frac{z^{-d}}{d_1\sqrt{d_1}}\sum_{j=0}^{d_1-1}\sum_{k=0}^{d_1-1}f(z_k)z_k^d \left(\frac{z}{z_k}\right)^j 
 \nonumber \\
 &= \frac{f(z)}{\sqrt{d_1}},
\end{align}
where in the last line we used Lemma~\ref{lem:interpolation1}.
\end{proof}
A block encoding for $f(U)$ in terms of $\mathcal{W}_U$ 
can then be obtained easily. 
For simplicity, we set $d_1 = 4d = 2^{m+2}$. Then the block encoding 
of $f(V_{4d})$ can be obtained 
using quantum Fourier transform (QFT)~\cite{NC00},
which follows from 
\begin{align}
    f(V_{4d}) &= \sum_{k'=0}^{4d-1}\ket{\varphi_{k'}}\bra{k'}
    \sum_{l=0}^{4d-1}f(z_l)\ket{l}\bra{l}
    \sum_{k=0}^{4d-1}\ket{k}\bra{\varphi_k} \nonumber\\
     &= 
    \text{QFT}
    \left(\braket{0|U_{f,4d}|0}\right)
    \text{QFT}^\dagger.
\end{align}
As $\mathcal{W}_U$ can be implemented by $4d-1$ queries to 
c-$U$, this block encoding uses ${\bf O}(d)$ 
queries to c-$U$, which is encouraging. 

However, the scaling factor
for this block encoding of $f(U)$ is $\sqrt{4d}$, which is suboptimal. 
This scaling factor can be improved to a constant by 
using a different interpolating polynomial given in the next lemma.
\begin{lemma}
\label{lem:laurentpolyapprox}
Let $f \in \mathbb{C}[z,z^{-1}]$ be a degree-$d$ Laurent polynomial and
$\{z_k = \exp\left(2\pi ik/4d\right), \ k=0,\dots,4d-1\}$. Then 
\begin{equation}
\label{eq:polyapprox}
        f(z) = \frac{1}{8d^2}\sum_{j'=d}^{3d-1}\sum_{j,k=0}^{4d-1} f(z_k)\left(\frac{z}{z_k}\right)^{j-j'}.
\end{equation}
\end{lemma}
\begin{proof}
Observe that for all $j' \in \{d,d+1,\dots,3d-1\}$, we have
\begin{equation}
\label{eq:polyapproxproof}
        f(z) = \frac{1}{4d}\sum_{j,k=0}^{4d-1} f(z_k)\left(\frac{z}{z_k}\right)^{j-j'}.
\end{equation}
We obtain Eq.~\eqref{eq:polyapprox} by averaging 
Eq.~\eqref{eq:polyapproxproof} over $j' \in \{d,d+1,\dots,3d-1\}$.
\end{proof}
As a corollary of Lemma~\ref{lem:laurentpolyapprox}, we have 
\begin{equation}
        f(U) = \frac{1}{8d^2}\sum_{j'=d}^{3d-1}\sum_{j,k=0}^{4d-1} 
        f(z_k)\left(\frac{U}{z_k}\right)^{j-j'},
\end{equation}
which follows by linearity.

Based on this interpolation, we can now construct a block encoding of $f(U)$ with
constant scaling factor.
\begin{lemma}
Let $f \in \mathbb{C}[z,z^{-1}]$ be a degree-$d$ Laurent polynomial and
\begin{equation}
    \ket{+_{2d}} = \frac{1}{\sqrt{2d}}\sum_{j=d}^{3d-1}\ket{j}, 
    \quad \ket{+_{4d}} = \frac{1}{\sqrt{4d}}\sum_{j=0}^{4d-1}\ket{j}.
\end{equation}
Let $U$, $\mathcal{W}_U$ and $V_{4d}$ be as defined in Lemma~\ref{lem:blockencoding1}.
Then
\begin{equation}
    f(U)= \sqrt{2}\braket{+_{2d}|\mathcal{W}_U^\dagger [\mathds{1}\otimes 
    f(V_{4d})]\mathcal{W}_U|+_{4d}}.
\end{equation}
\end{lemma}

\begin{proof}
A similar calculation to the one in the proof of Lemma~\ref{lem:blockencoding1} yields
\begin{align}
\label{eq:runawaycircuit}
        \sqrt{2}\bra{+_{2d}}&\mathcal{W}_z^\dagger f(V_{4d})\mathcal{W}_z\ket{+_{4d}}\nonumber\\
        &= \frac{1}{8d^2}\sum_{j'=d}^{3d-1}\sum_{j,k=0}^{4d-1} 
        f(z_k)\left(\frac{z}{z_k}\right)^{j-j'}\nonumber\\
        &=f(z). 
\end{align}
By linearity, we have
\begin{align}
\label{eq:runawaycircuit}
        \sqrt{2}\bra{+_{2d}}&\mathcal{W}_U^\dagger f(V_{4d})\mathcal{W}_U\ket{+_{4d}}\nonumber\\
        &= \frac{1}{8d^2}\sum_{j'=d}^{3d-1}\sum_{j,k=0}^{4d-1} 
        f(z_k)\left(\frac{U}{z_k}\right)^{j-j'}\nonumber\\
        &=f(U). 
\end{align}
QED.
\end{proof}

A unitary circuit for block encoding of $f(U)$ is shown in Fig.~\ref{fig:QSP}.
Interestingly, this circuit is remarkably similar to the circuit for function implementation based on 
quantum phase estimation, such as the one used in HHL algorithm
for matrix inversion. The circuit constructed above can be divided into three parts, 
the first part being the standard quantum phase estimation (QPE), 
the second part being the function implementation using $U_{f,4d}$, and
the third part resembling uncomputing QPE with one of the Hadamard gates
replaced by a zero-controlled Not gate. This change of a single gate has a huge impact 
on the performance of the overall circuit! The analysis in Ref.~\cite{HHL09} showed that
in the absence of this small replacement, the error in implementation
scales as $\text{poly}(1/d)$. This is because 
the first part of the circuit, which performs phase estimation,
records phases of $U$ to an accuracy ${\bf O}(1/d)$. 
Yet, our analysis shows that as long as $f$ is a degree-$d$ Laurent polynomial,
the circuit in Fig.~\ref{fig:QSP} block encodes $f(U)$ exactly!

So far, we have constructed an exact block encoding of 
$f(U)$ assuming $f$ to be a degree-$d$ Laurent polynomial.
Our strategy to address the case of general $f$, and thereby to prove 
Theorem~\ref{thm:newQSP}, is to use the exact same
circuit to block encode $f(U)$. Of course, this block encoding is not exact,
and in order to prove Theorem~\ref{thm:newQSP}, we need
to bound the error.

\begin{proof}[Proof of Theorem~\ref{thm:newQSP}]
We first prove the bound on the error. 
The operator exactly block encoded
by our circuit is 
\begin{align}
\label{eq:deffd}
    \sqrt{2}\bra{+_{2d}}&\mathcal{W}_U^\dagger 
f(V_{4d})\mathcal{W}_U\ket{+_{4d}} \nonumber\\
&= \frac{1}{8d^2}\sum_{j'=d}^{3d-1}\sum_{j,k=0}^{4d-1} f(z_k)
        \left(\frac{U}{z_k}\right)^{j-j'}\nonumber\\ 
        &= f_d(U),
\end{align}
where we introduced
\begin{equation}
\label{eq:laurentpolyapprox}
	f_d(z)= \frac{1}{8d^2}\sum_{j'=d}^{3d-1}\sum_{j,k=0}^{4d-1} f(z_k)
        \left(\frac{z}{z_k}\right)^{j-j'}.
\end{equation}
Note that $f_d$ is a degree-$(3d-1)$ Laurent polynomial.
Our task now is to establish an upper bound on $\norm{f(U)-f_d(U)}$.

Using subadditivity of the operator norm, we have
\begin{multline}
\label{eq:triangleinequality}
\norm{f(U)-f_d(U)}\le\\ \norm{f(U)-f_{d,*}(U)} + \norm{f_{d,*}(U)-f_d(U)},
\end{multline}
where $f_{d,*}$ is the best degree-$d$ Laurent polynomial approximation to $f$ as before.
A bound on the first term on the right-hand side is obtained using
\begin{equation}
\norm{f(U)-f_{d,*}(U)} \le \norm{f-f_{d,*}}_\infty = E_d(f).
\end{equation}
To bound $\norm{f_d(U)-f_{d,*}(U)}$, observe that
\begin{equation}
    f_d(z)-f_{d,*}(z)= \sqrt{2}\braket{+_{2d}|\mathcal{W}_z^\dagger\, 
    r_d(V_{4d})\mathcal{W}_z|+_{4d}},
\end{equation}
with $r_d  = f - f_{d,*}$. Since $\mathcal{W}_z$ is unitary for $\abs{z}=1$,
we have $\norm{\mathcal{W}_z}=1$, which leads to 
\begin{align}
    \abs{f_d(z)-f_{d,*}(z)} &\le \sqrt{2}\abs{f(z) - f_{d,*}(z)} \nonumber\\
    \implies \norm{f(U)-f_d(U)} &\le \norm{f_{d,*}-f_d}_{\infty} \le \sqrt{2}E_d(f).
\end{align}
Now by Eq.~\eqref{eq:triangleinequality}, we have 
\begin{equation}
    \norm{f(U)-f_d(U)} \le (1+\sqrt{2})E_d(f),
\end{equation}
as stated in the theorem.

We now count the resources. The circuit in Fig.~\ref{fig:QSP} 
uses $m+3$ ancilla qubits and makes 
$4d-1$ uses of c-$U$ and c-$U^\dagger$ each and one use 
of $U_{f,4d}$. QFT acting on $m+2$ qubits requires 
${\bf O}(m^2)$ additional single- and two-qubit gates~\cite{NC00}. Finally,
generating a description of the circuit using $\text{poly}(m)$ classical
operations is straightforward, which completes the proof.
\end{proof}

It is worth remarking that some simple strategies can be employed to
improve the prefactor for the uses of c-$U$ and c-$U^\dagger$. For instance,
instead of constructing a block encoding of $f(U)$ based on Eq.~\eqref{eq:deffd},
one may use the equivalent expression
\begin{equation}
        f_d(U) =  \sqrt{2}\braket{+_{2d}|(\mathcal{W}''_U)^\dagger 
    f(V_{4d})\mathcal{W}'_U|+_{4d}}
\end{equation}
as the starting point, where 
\begin{equation}
    \mathcal{W}'_U= \sum_{j=-2d}^{2d-1}U^j \otimes \ket{j}\bra{j}, \quad 
    \mathcal{W}''_U = \sum_{j=-d}^{d-1}U^j \otimes \ket{j}\bra{j}.
\end{equation}
Then $\mathcal{W}'_U$ ($\mathcal{W}''_U$) can be implemented with $2d$ ($d$) 
uses of c-$U$ and c-$U^\dagger$ each. As a result, the block encoding of $f(U)$
makes $3d$ uses of c-$U$ and c-$U^\dagger$ each. One can achieve further savings by
using directionally-controlled operation 
$W = \ket{0}\bra{0}\otimes U+ \ket{1}\bra{1}\otimes U^\dagger$, which roughly 
halves the query complexity. It is also possible to choose a smaller value for 
$d_1$ instead of $d_1=4d$ used above, at the cost of a larger scaling
constant for the block encoding of $f(U)$.

To perform QSP according to the circuit constructed above, 
the unitary $U_{f,4d}$ needs to be constructed starting from a description 
of $f$. There are several ways to construct such a unitary~\cite{ZNSD24}. 
One approach~\cite{HHL09}
starts with an efficient reversible classical circuit $O_{f}$ 
for computing $f$ to $b$ bits of precision. A quantum circuit for 
computing $f$ can then be constructed by replacing each classical gate
by its quantum counterpart. Now an approximation $\tilde{U}_{f,4d}$
to $U_{f,4d}$ is obtained by first computing 
$f$ using $O_{f,4d}$ and recording the outcome in the $b$-bit register, 
then applying controlled-rotation by an angle written in this register and
finally uncomputing the $b$-bit register by applying $O_f^\dagger$. In effect,
$\tilde{U}_{f,4d}$ makes one use of $O_f$ and $O_f^\dagger$ each 
and ${\bf O}(b)$ two-qubit gates. The error $\delta = \norm{\tilde{U}_{f,4d}-U_{f,4d}}$ 
depends on the binary representation of complex numbers being used, but for commonly used
representation, this error scales as $\delta \in {\bf O}(2^{-b})$. Finally $U_{f,4d}$
in Theorem~\ref{thm:newQSP} can be replaced by $\tilde{U}_{f,4d}$, leading to a total 
error bounded by $(1+\sqrt{2})E_d(f)+\sqrt{2}\delta$.

Finally, observe that in the process of proving Theorem~\ref{thm:newQSP},
we have discovered an interesting Laurent polynomial approximation
to a given function on the unit circle, as we state in the next theorem.
\begin{theorem}
\label{thm:laurentpolyapprox}
    For any $f:S^1 \to \mathbb{C}$, the polynomial 
    approximation $f_d: S^1 \to \mathbb{C}$ defined in 
    Eq.~\eqref{eq:laurentpolyapprox} satisfies 
\begin{equation}
    \norm{f-f_d}_\infty \le (1+\sqrt{2})E_d(f).
\end{equation}
\end{theorem}
\noindent Quite interestingly, the approximation error is bounded by a 
constant times the approximation error of the best Laurent
polynomial approximation. However, we emphasize that we are comparing
the approximation error of a higher degree polynomial with that of the best
approximation error achievable by lower-degree polynomials;
$f_d$ is a degree-$(3d-1)$ Laurent polynomial and achieves an approximation 
error of the order of $E_d(f)$. Nevertheless,
this polynomial approximation could be useful in situations where
one requires error scaling to be provably equal to the error scaling of
the best polynomial approximation.

This concludes our discussion of the new circuit for QSP. In the next section, we explore
its applications to implementing functions of a block-encoded
Hermitian matrix and singular value transformation of an arbitrary block-encoded matrix.

\section{Applications to functions of a Hermitian matrix and singular value transformation}
\label{sec:qevt}

As in the case of conventional QSP, this new circuit for QSP can be 
used for implementing a given function $h:[-1,1]\to \mathbb{C}$ 
of a Hermitian matrix $H$ block-encoded in 
a unitary $U_H$. We first recall a known lemma and include its proof
for completeness~\cite{GSLW19}.
\begin{lemma}
\label{lem:arccos}
Let $U_H$ be an $(n+a)$-qubit self-inverse block encoding of an $n$-qubit 
Hamiltonian $H$, and let $G$ be the Hermitian matrix with spectrum in $[-\pi,\pi)$
satisfying
\begin{equation}
e^{iG} = (2\Pi-\mathds{1})U_H, \quad \Pi = \ket{0^a}\bra{0^a}.
\end{equation}
Then for any function $g:[-1,1]\to \mathbb{C}$, 
    \begin{equation}
        g(H) = \braket{0^a|g(\cos{G})|0^a}.
    \end{equation}
\end{lemma}

\begin{proof}
    It suffices to show that for any eigenstate $\ket{\lambda}$ of $H$, we have 
    $\braket{0^a| g(\cos{G})|0^a}\ket{\lambda} = g(\lambda)\ket{\lambda}$. To prove
    this equality, let
    \begin{equation}
        \ket{\psi} = \ket{0^a}\ket{\lambda},\quad
        \ket{\psi^\perp} = (\mathds{1}-\Pi)U_H\ket{0^a}\ket{\lambda}/\sqrt{1-\lambda^2}.
    \end{equation}
    Then
    \begin{align}
        U_H\ket{\psi} &= \lambda \ket{\psi} + \sqrt{1-\lambda^2} \ket{\psi^\perp},\nonumber\\
        U_H\ket{\psi^\perp} &= \sqrt{1-\lambda^2} \ket{\psi} -  \lambda\ket{\psi^\perp}.
    \end{align}
    Using 
    \begin{align}
        (2\Pi-\mathds{1})\ket{\psi} = \ket{\psi}, \quad
        (2\Pi-\mathds{1})\ket{\psi^\perp} = -\ket{\psi^\perp},
    \end{align}
    we get
    \begin{align}
        (2\Pi-\mathds{1})U_H\ket{\psi} &= \lambda \ket{\psi} 
        - \sqrt{1-\lambda^2} \ket{\psi^\perp},\nonumber\\
        (2\Pi-\mathds{1})U_H\ket{\psi^\perp} &= \sqrt{1-\lambda^2} \ket{\psi} +  \lambda\ket{\psi^\perp}.
    \end{align}
    Therefore, the space $\mathcal{V} = \text{span}\{\ket{\psi},\ket{\psi^\perp}\}$
    is invariant under the action of $(2\Pi-\mathds{1})U_H$, and furthermore, 
    $(2\Pi-\mathds{1})U_H$ acts as 
    \begin{align}
        [(2\Pi-\mathds{1})U_H]_{\mathcal{V}} &= 
        \begin{bmatrix}
            \lambda & -\sqrt{1-\lambda^2}\\
            \sqrt{1-\lambda^2} & \lambda
        \end{bmatrix} \nonumber \\
        &= \exp(-i\arccos(\lambda)\sigma_y)
    \end{align}
    in this basis. Then $G$ acts on $\mathcal{V}$ as
    \begin{equation}
        [G]_{\mathcal{V}} = -\arccos(\lambda)\sigma_y.
    \end{equation}
    Consequntly, 
    \begin{equation}
        g(\cos{[G]})_{\mathcal{V}} = g(\cos(\arccos(\lambda))\mathds{1}) = g(\lambda)\mathds{1},
    \end{equation}
    which proves $g(\cos{[G]})\ket{\psi} = g(\lambda)\ket{\psi}$ as desired.
\end{proof}

The proof of Theorem~\ref{thm:newFHM} is now straightforward.
\begin{proof}[Proof of Theorem~\ref{thm:newFHM}]
Notice that $(2\Pi-\mathds{1})$ can be implemented using ${\bf O}(a)$ two-qubit gates.
Since $f(e^{i\theta}) = g(\cos{\theta})$, we have 
\begin{equation}
f\left((2\Pi-\mathds{1})U_H\right) = f(e^{iG}) = g(\cos{G}),
\end{equation}
which is a block encoding of $g(H)$. 
Using Theorem~\ref{thm:newQSP}, 
a $(\sqrt{2},a+m+3,(1+\sqrt{2})E_d(f))$-block encoding of 
$f\left((2\Pi-\mathds{1})U_H\right)$ 
can be constructed using $U_{f,4d}$. 
The number of queries to $U_H$ and $U_H^\dagger$ as well as the
number of classical operations follow from Theorem~\ref{thm:newQSP}.
\end{proof}

In comparison, the conventional QSP can achieve precision
up to $E_d(f)$ with $d$ queries to $U$, but requires
$\text{poly}(d) = \exp(m)$ classical operations for finding 
the best polynomial approximation and rotation angles to
implement the same. In this sense, our algorithm achieves
nearly optimal queries with exponentially smaller classical cost.

Before moving to the applications of Theorem~\ref{thm:newFHM}, 
let us look more closely at the circuit for FHM. This circuit
constructs a block encoding for an approximation of $f(e^{iG})$, 
namely $f_d(e^{iG})$, 
thereby implementing the approximation  
$g_d(\cos{G}) = f_d(e^{iG})$ to $g(\cos{G})$. 
It is insightful to obtain an explicit expression for $g_d$ 
Using $f(e^{i\theta}) = g(\cos{\theta})$ in Eq.~\eqref{eq:laurentpolyapprox}, we have
\begin{align}
	f_d(z)&= \frac{1}{8d^2}\sum_{j'=d}^{3d-1}\sum_{j,k=0}^{4d-1} 
    f(z_k)
        \left(\frac{z}{z_k}\right)^{j-j'} \nonumber\\
        &= \frac{1}{8d^2}\sum_{j'=d}^{3d-1}\sum_{j,k=0}^{4d-1} 
    g(x_k)
        \left(\frac{z}{z_k}\right)^{j-j'},
\end{align}
where we defined $x_k = \cos{\theta_k}$ and $\theta_k = \arg{z_k}$.
We therefore have
\begin{equation}
\label{eq:1}
	f_d(e^{i\theta})
        = \frac{1}{8d^2}\sum_{j'=d}^{3d-1}\sum_{j,k=0}^{4d-1} 
    g(x_k)\left(\frac{e^{i\theta}}{z_k}\right)^{j-j'}.
\end{equation}
By changing the variables according $j \mapsto 4d-1-j$, $j' \mapsto 4d-1-j'$ and
$k \mapsto 4d-k \mod 4d$, we obtain
\begin{equation}
\label{eq:2}
    f_d(e^{i\theta}) = \frac{1}{8d^2}\sum_{j'=d}^{3d-1}\sum_{j,k=0}^{4d-1}g(x_k)
        \left(\frac{e^{-i\theta}}{z_k}\right)^{j-j'} = f_d(e^{-i\theta}).
\end{equation}
Since we have discovered that $f_d(e^{i\theta}) = f_d(e^{-i\theta})$, we now
define a new function $g_d:[-1,1]\to\mathbb{C}$ by $g_d(\cos{\theta}) 
= f_d(e^{i\theta})$.
By averaging Eqs.~\eqref{eq:1} and \eqref{eq:2}, we obtain
\begin{align}
    g_d(\cos{\theta}) &= \frac{1}{8d^2}\sum_{j'=d}^{3d-1}\sum_{j,k=0}^{4d-1}g(x_k)
        \left(\frac{e^{i(j-j')\theta}+e^{-i(j-j')\theta}}{2z_k^{j-j'}}\right) \nonumber\\
        &= \frac{1}{8d^2}\sum_{j'=d}^{3d-1}\sum_{j,k=0}^{4d-1}g(x_k)
        \frac{\cos((j-j')\theta)}{z_k^{j-j'}}
        \nonumber\\
        &= \frac{1}{8d^2}\sum_{j'=d}^{3d-1}\sum_{j,k=0}^{4d-1}g(x_k)
        \frac{T_{\abs{j-j'}}(\cos{\theta})}{z_k^{j-j'}},
\end{align}
where $T_j$ denotes Chebyshev polynomial of the first kind of order $j$.
We have therefore revealed that our circuit for FHM implements the 
degree-$(3d-1)$ polynomial approximation to $g(x)$ given by
\begin{equation}
\label{eq:polyapprox2}
    g_d(x) = \frac{1}{8d^2}\sum_{j'=d}^{3d-1}\sum_{j,k=0}^{4d-1}g(x_k)
        \frac{T_{\abs{j-j'}}(x)}{z_k^{j-j'}}.
\end{equation}
With further algebra, $g_d(x)$ can also be expressed as
\begin{equation}
        g_d(x)=\sum_{r=0}^{3d-1}\beta_r T_r(x), 
\end{equation}
where 
\begin{equation}
    \beta_r = \left\{ \begin{array}{lcl}
    \frac{1}{4d}\sum_{k=0}^{4d-1}g(x_k) & \text{if} & r=0\\
    \frac{2}{4d}\sum_{k=0}^{4d-1}g(x_k)T_r(x_k) & \text{if} & 1 \le r \le d\\
    \frac{2(3d-r)}{8d^2}\sum_{k=0}^{4d-1}g(x_k)T_r(x_k) & \text{if} & d < r \le 3d-1.
    \end{array}\right.
\end{equation}
Theorem~\ref{thm:newFHM} also bounds the approximation error by 
$(1+\sqrt{2})E_d(f)$, where $f(e^{i\theta}) = g(\cos{\theta})$
or equivalently $f = g \circ \cos \circ \arg$. 
We have therefore obtained a result on polynomial approximation of a function,
which we now state in the form of a theorem.
\begin{theorem}
\label{thm:polyapprox}
    Let $g:[-1,1]\to\mathbb{C}$ be a function and let $g_d:[-1,1]\to\mathbb{C}$ be
    defined by Eq.~\eqref{eq:polyapprox2}. Then for any $x \in [-1,1]$,
    \begin{equation}
\label{eq:polyapproxerror}
    \abs{g(x)-g_d(x)} \le (1+\sqrt{2})E_d(g \circ \cos \circ \arg).
\end{equation}
\end{theorem}
\noindent Theorem~\ref{thm:polyapprox} can also be proved directly using
Theorem~\ref{thm:laurentpolyapprox}.
The polynomial approximation $g_d$ can be used for performing FHM using 
conventional QSP circuit based on alternating projections. 
Such an approach would guarantee an error scaling
proportional to the best polynomial approximation. Interestingly, 
calculation of $g_d$ does not require any iterative procedure 
like for Remez exchange algorithm. Nevertheless, the calculation of the coefficients of 
$g_d$ requires $\text{poly}(d)$ classical operations, which is not desirable. 
A bigger disadvantage of this approach is that
we need a degree-$(3d-1)$ polynomial $g_d$ to obtain an error of the order
of $E_d(f)$. In practice, a degree-$(3d-1)$ polynomial approximation obtained by
truncating the Chebyshev series of $g$ fares significantly better 
than $g_d$ for most functions.

Returning to the applications of Theorem~\ref{thm:newFHM}, let us consider the task of
Hamiltonian simulation. We are given a blcok encoding $U_H$ of
a Hamiltonian $H$, a time $t \in \mathbb{R}^+$, a precision $\epsilon$,
and the task is to construct a block encoding of $e^{iHt}$. Clearly,
we have $g(x) = e^{itx}$ and therefore $f(\theta) = e^{it\cos{\theta}}$.
In Ref.~\cite{GSLW19}, it was shown that there exists a degree-$d$ trigonometric polynomial
approximating $f(\cos{\theta})$ such that the approximation error $\varepsilon_d(f)$ is bounded
from above by 
\begin{equation}
	\varepsilon_d(f) \le \frac{5}{4}\left(\frac{e\abs{t}}{2d}\right)^d, \quad d \ge \abs{t}-1.
\end{equation}
Since $E_d(f)$ is the approximation error for the best trigonometric polynomial, we have 
$E_d(f) \le \varepsilon_d(f)$. Consequently, our circuit generates a block encoding of $e^{iHt}$
with error 
\begin{equation}
	\epsilon \le \frac{(1+\sqrt{2})5}{4}\left(\frac{e\abs{t}}{2d}\right)^d
\end{equation}
using ${\bf O}(d)$ queries to $U_H$ and $U_H^\dagger$. It is easy to verify following the steps
in Ref.~\cite{GSLW19} that this error scaling leads to an optimal algorithm for Hamiltonian simulation.
Similar arguments could be extended for implementation of
other functions of Hermitian matrices.

Theorem~\ref{thm:newFHM} provides an error bound $\epsilon = (1+\sqrt{2})E_d(f)$
for a given $d$. This result can be bolstered with a general, tight upper bound
on $E_d(f)$ that depends on some properties of $f$. Fortunately, such tight bounds
on $E_d(f)$ are known in the form of Jackson's inequalities~\cite{Ste05} and their
generalizations for a large class of functions. These inequalities bound $E_d(f)$
in terms of some properties of the function $\tilde{f}:\mathbb{R} \to \mathbb{C}$
defined by $\tilde{f}(\theta) = f(e^{i\theta}) = g(\cos{\theta})$.
For instance, if $\tilde{f}$ is $r$ times differentiable, 
then Jackson's inequalities bound $E_d(f)$ in terms of the modulus of
continuity of the $r$th derivative of $\tilde{f}$. 
\begin{definition}
For a function $\tilde{f}: \mathbb{R} \to \mathbb{C}$ and a $\Delta \in \mathbb{R}^+$, 
the modulus of continuity $\omega(\Delta,\tilde{f})$ is defined to be
\begin{equation}
    	\omega(\Delta,\tilde{f}) = \sup_{\abs{x_1-x_2} \le \Delta}
        \left\{\abs{\tilde{f}(x_1)-\tilde{f}(x_2)} \right\}.
\end{equation}
\end{definition}
\noindent Jackson's inequality gives us the following bound on the best approximation error $E_d(f)$.
\begin{lemma}
\label{lem:jackson}
    Let $f:S^1 \to \mathbb{C}$ be a function such that $\tilde{f}(\theta) = f(e^{i\theta})$
    is $r$ times differentiable. Then 
    \begin{equation}
        E_d(f) \le \frac{C(r)\omega(1/d,\tilde{f}^{(r)})}{d^r}, 
    \end{equation}
    where
    \begin{equation}
        C(r) = \frac{4}{\pi}\sum_{k=0}^{\infty}\frac{(-1)^{k(r+1)}}{(2k+1)^{r+1}}
    \end{equation}
    is a constant independent of $f$ and $d$.   
\end{lemma}
\noindent The constants $\{C(r)\}$ are known as Akhiezer-Krein-Favard constants. For $r=0$ and $r=1$, we have
$C(0)=1$ and $C(1)=\pi/2$ respectively.

As a specific application, we now derive the complexity of implementing
$\abs{H}^c$ for $c \in \mathbb{R}^+$. It suffices to derive a bound on $E_d(f)$
for the corresponding function $f$.
\begin{theorem}
    Let $c$ be a non-integer positive constant, i.e. $c \in \mathbb{R}^+$, $c \notin \mathbb{Z}^+$. 
    Let $g(x) = |x|^c$ and $f(e^{i\theta}) = g(\cos{\theta}) = \abs{\cos{\theta}}^c$.
    Then 
    \begin{equation}
        E_d(f) \in  {\bf O}\left(\frac{1}{d^c}\right).
    \end{equation}
\end{theorem}

\begin{proof}
Note that $f$ is $\lfloor c \rfloor$ times differentiable for all $\theta \in \mathbb{R}$,
whereas its $(\lfloor c \rfloor+1)$th derivative diverges at the values of $\theta$ for which $\cos(\theta) = 0$. 
To apply Jackson's inequality in Lemma~\ref{lem:jackson}, we need
to bound $\omega(1/d, f^{\lfloor c \rfloor})$. Suppose first that $\lfloor c \rfloor$ is an even integer.
Then we can express $f^{\lfloor c \rfloor}(\theta)$ as 
\begin{equation}
f^{\lfloor c \rfloor}(\theta) = 
\sum_{j = 0}^{\lfloor c \rfloor/2} \alpha_j(c) \abs{\cos{\theta}}^{c-2j}
\end{equation} 
for some constants $\alpha_j(c) \in \mathbb{R}$. Then 
\begin{multline}
\omega(1/d, f^{\lfloor c \rfloor}) \le \omega(1/d, \sum_{j = 0}^{\lfloor c \rfloor/2-1} \alpha_j(c) 
\abs{\cos{\theta}}^{c-2j})
+ \\
\omega(1/d, \alpha_{\lfloor c \rfloor/2}(c) \abs{\cos{\theta}}^{c-\lfloor c \rfloor}).
\end{multline}
Out of the two terms on the right-hand-side, the first is the modulus of continuity
of a function with derivative of bounded from above, say by $K_{c,1}$. Therefore
\begin{align}
\omega(1/d, \sum_{j = 0}^{\lfloor c \rfloor/2-1} \alpha_j(c) 
\abs{\cos{\theta}}^{c-2j}) &\le K_{c}/d \nonumber\\
&\le K_{c}/d^{c-\lfloor c \rfloor}.
\end{align}
Now let us consider the second term. The derivative of $\abs{\cos{\theta}}^{c-\lfloor c \rfloor}$
diverges at the values of $\theta$ 
for which $\cos(\theta) = 0$. To bound the modulus of continuity, we use
\begin{align}
    \omega(1/d, \alpha_{\lfloor c \rfloor/2}(c) \abs{\cos{\theta}}^{c-\lfloor c \rfloor}) 
  &= \alpha_{\lfloor c \rfloor/2}(c)\abs{\sin{\frac{1}{d}}}^{c-\lfloor c \rfloor} \nonumber\\
  &\le \frac{\alpha_{\lfloor c \rfloor/2}(c)}{d^{c-\lfloor c \rfloor}}.
\end{align}
Putting together the two terms, we get
\begin{equation}
\omega(1/d, f^{\lfloor c \rfloor}) \le \frac{K'_c}{d^{c-\lfloor c \rfloor}}, \quad 
K'_c = K_c + \alpha_{\lfloor c \rfloor/2}(c).
\end{equation}
For the case $\lfloor c \rfloor$ odd, we can use the expansion
\begin{equation}
f^{\lfloor c \rfloor}(\theta) = 
\sum_{j = 0}^{(\lfloor c \rfloor-1)/2} \alpha_j(c) \abs{\cos{\theta}}^{c-2j+1}\text{sign}(\cos{\theta})\sin{\theta},
\end{equation} 
to show that
\begin{equation}
\omega(1/d, f^{\lfloor c \rfloor}) \le \frac{K''_c}{d^{c-\lfloor c \rfloor}}
\end{equation}
for a $c$-dependent constant $K''_c$.
Now using Lemma~\ref{lem:jackson}, we obtain
\begin{equation}
E_d(f) \le \frac{\omega(1/d, f^{\lfloor c \rfloor})}{d^{\lfloor c \rfloor}} \le \frac{K}{d^c},
\end{equation}
where $K = K'_c$ for $\lfloor c \rfloor$ even and $K = K''_c$ for $\lfloor c \rfloor$ odd. QED.
\end{proof}

Observe that $E_d(f)$ for the functions considered above does not decay exponentially in $d$. 
This is not a limitation of our approach, rather an unavoidable consequence of the fact
that $g(\theta) = \abs{\cos{\theta}}^c$ is not analytic at $\theta = \pm \pi/2$ when $c$ is not an integer. 
The query complexity for implementing 
$\abs{H}^c$ using our approach is optimal up to the prefactor
and scales as ${\bf O}(1/\epsilon^{1/c})$. As a corollary, a block encoding
of $\sqrt{\abs{H}}$ can be constructed using ${\bf O}(1/\epsilon^2)$ queries to $U_H$ and $U_H^\dagger$.

Stronger bounds on the scaling of $E_d(f)$ are known for classes of functions 
with a higher degree of smoothness~\cite{Ste05}. 
Such classes include infinitely differentiable functions,
functions analytic in a neighborhood of the real line and for functions
that are analytic everywhere, i.e. entire functions. 
These bounds are stated in terms of $(\psi,\beta)$-derivative framework~\cite{Ste86}
used for classification of smooth functions.
Qualitatively, these results show that $E_d(f)$ decays exponentially in $d$
if $\tilde{f}(\theta) = f(e^{i\theta})$ is analytic in a neighborhood of the real line. 
Moreover, if $\tilde{f}(\theta)$ is entire, then 
$E_d(f)$ scales super-exponentially with $d$, i.e. for any $\gamma>0$, there exists
a $d_0 \in \mathbb{Z}^+$ such that $E_d(f) \le \exp(-\gamma d)$ for $d>d_0$~\cite{SSS08,Ste05}.
Interestingly, analyticity of $g$ implies that of $\tilde{f} = g(\cos{\theta})$. Therefore, 
our circuit for implementing $g(H)$ achieves super-exponential accuracy
for any entire function $g$. Such results were previously proved in the literature 
for some specific functions on a case-by-case basis~\cite{GL17,GSLW19,AGGW20}. The examples
include $g(x) = e^{ixt}$ appearing in Hamiltonian simulation~\cite{AT03,BCK15},
$g(x) = (1-e^{ixt})/x$ appearing in differential-equation solvers~\cite{Ber14,Kro23} and
$g(x)=e^{-\beta x}$ appearing in Gibbs sampling~\cite{PW09,HMSS+22} and SDP solvers~\cite{AGGW20}.

Theorem~\ref{thm:newFHM} can be further used for
performing quantum singular value transformation (QSVT) using known techniques. 
In QSVT, we are given a block encoding of an arbitrary matrix $A$ with
singular value decomposition $A=V_1^\dagger S V_2$ and a function
$g$ of definite parity, and the goal is to construct a block encoding of $g^{\text{SV}}(A)$, 
where $g^{\text{SV}}(A) = V_1^\dagger g(S) V_2$ if $g$ is an odd function and
$g^{\text{SV}}(A) = V_2^\dagger g(S) V_2$ if $g$ is an even function.
To perform QSVT, we first encode $A$ in the off-digonal block of a Hermitian matrix 
\begin{equation}
H_A = \begin{bmatrix}
0 & A^\dagger \\ A & 0
\end{bmatrix}.
\end{equation}
A block encoding of $H_A$ can be constructed by making one use of the block encoding of
$A$ and its inverse respectively. Then we construct a block encoding of 
$g(H_A)$ to desired accuracy as described in Theroem~\ref{thm:newFHM},
We then invoke the fact that~\cite{GSLW19}
\begin{equation}
    g^{\rm SV}(A) = \left\{\begin{array}{lcl}
    \braket{0|g(H_A)|0} & \text{if} & p=0\\
    \braket{1|g(H_A)|0} & \text{if} & p=1 
    \end{array}\right. \quad .
\end{equation}
Therefore, our new circuit for QSP enables us to achieve
nearly optimal dependence on precision for
both FHM and QSVT for arbitrary $g$, 
without expensive or unstable classical preprocesing.

\section{Conclusion}
\label{sec:conclusion}
In conclusion, we have constructed a quantum circuit for QSP based on interpolation
by Laurent polynomials. This circuit overcomes the bottleneck of the existing approaches, namely 
expensive and unstable classical preprocessing. 
The circuit is quite simple and universal with only one tunable parameter $d$, 
which determines the precision. Our circuit for QSP enables achieving nearly optimal scaling 
for implementing functions of block-encoded Hermitian matrices
and for performing singular value transformation of arbitrary block-encoded matrices 
with circuits that are easy to construct classically.
It is particularly valuable when one is interested in implementing black-box functions,
as only a constant number of queries to the function oracle are required.

Interestingly, our circuit is very similar to
the phase-estimation based circuit for function implementation, and 
therefore uncovers a new connection between quantum phase estimation and
QSP. Prior to this work, phase estimation-based approaches for quantum signal processing and
for implementing functions of Hermitian matrices were considered ineffective,
as they were assumed to lead to polynomial precision,
in contrast to exponential precision offered by alternating projection 
methods~\cite{LC17,GSLW19}. 
We show that an exponential precision can be achieved by phase estimation-based
approaches by replacing only one Hadamard gate in the circuit by a two-qubit gate, 
as explained in Sec.~\ref{sec:qsp}.

In the future, we shall investigate whether our techniques can be applied  
for quantum eigenvalue transformation for non-normal
matrices~\cite{LS24,ACLY24}. Building on the insights from this work, it would be interesting
to see whether the techniques for phase estimation for non-normal
matrices~\cite{Sha22,LS24,AK24} prove useful for this task.

\begin{acknowledgments}
This research was supported by the 
Australian Research
Council Centre of Excellence for Engineered Quantum Systems (CE170100009). 
I am grateful for insightful conversations with Salini Karuvade, 
Gopikrishnan Muraleedharan and Alessandro Luongo.
\end{acknowledgments}

\bibliography{Ref}

\end{document}